\documentclass[%
 reprint,
 amsmath,amssymb,
 aps,prx
]{revtex4-2}

\usepackage{amsmath}
\usepackage{bm}
\usepackage{amsfonts}
\usepackage[cal=cm]{mathalpha}
\usepackage{nicematrix}
\usepackage{graphicx}
\usepackage{tcolorbox}
\usepackage[colorlinks=true, allcolors=blue]{hyperref}
\usepackage{amsthm}
\usepackage{algorithm}
\usepackage{algpseudocode}
\usepackage{hyperref}
\usepackage{epstopdf}
\usepackage{cancel}

\newtheorem{theorem}{Theorem}[section]

\begin{document}

\title{Explosive opinion {spreading with polarization and depolarization via} asymmetric perception}

\author{Haoyang Qian} 
\author{Malbor Asllani} 
 \affiliation{Department of Mathematics, Florida State University,
1017 Academic Way, Tallahassee, FL 32306, United States of America}

\begin{abstract}
Polarization significantly influences societal divisions across economic, political, religious, and ideological lines. Understanding these mechanisms is key to devising strategies to mitigate such divisions and promote depolarization. Our study examines how asymmetric opinion perception, modeled through nonlinear incidence terms, affects polarization and depolarization within structured communities. We demonstrate that such asymmetry leads to explosive polarization and causes a hysteresis effect responsible for abrupt depolarization. We develop a mean-field approximation to explain how nonlinear incidence results in first-order phase transitions and the nature of bifurcations. This approach also helps {in understanding} how opinions polarize according to underlying social network communities {and how these phenomena intertwine with the nature of such transitions}. Numerical simulations corroborate the analytical findings.
\end{abstract}

\maketitle


In recent years, the unprecedented increase in data availability and rapid progress in network science have greatly enhanced the understanding of complex social phenomena by adopting principles from physics to model interactions within social systems \cite{jackson2008social, lazer2009computational, centola_spread_2010, cinelli2021echo, castellano_statistical_2009, easley2010networks, newman_networks_2018}. This interdisciplinary approach has been instrumental in studying various aspects of social dynamics, including cultural dissemination \cite{axelrod1997dissemination}, language evolution \cite{nowak2002computational}, and crowd behavior \cite{helbing2000simulating}, among others. Of these, the study of opinion dynamics within social networks has gained significant attention, driven by the increasing influence of social media and digital communication platforms \cite{bail2018exposure}. Polarization is a critical issue in opinion dynamics, where two or more groups within a population hold opposing and sometimes extreme views \cite{bail2018exposure, baumann2020modeling, baumann2021emergence, ojer_modeling_2023}. Such polarization can lead to significant societal impacts, such as increased political division \cite{bail2018exposure} and the formation of echo chambers \cite{baumann2020modeling}. Understanding how opinions form, evolve, and spread within social networks is crucial for addressing contemporary social challenges, leveraging digital communication to achieve positive social outcomes, and comprehending the mechanisms driving polarization to develop strategies for fostering a more cohesive society \cite{balietti2021reducing}.

Opinion dynamics exhibit explosive transitions as a distinct feature, characterized by sudden and large-scale changes in social systems. For instance, exposure to opposing views on social media can increase political polarization, leading to abrupt shifts in public opinion \cite{bail2018exposure}. Similarly, the adoption of behaviors in online social networks can experience rapid and widespread changes once a critical threshold is reached \cite{centola_spread_2010}. In financial markets, herding behavior and information cascades can result in sudden crashes, illustrating the explosive nature of market dynamics \cite{sornette2003crash}. Models of collective behavior further emphasize how small changes can lead to large-scale social movements \cite{granovetter1978threshold}{, while} global cascades in networks {show} how minor perturbations can trigger explosive transitions in social systems \cite{watts2002simple}.

{Contemporary platforms for information dissemination—such as online social networks, digital news platforms, and other media channels—play a critical role in creating an uneven perception of opinions \cite{bail2018exposure, baumann2020modeling}. On one side, the overwhelming exposure to such platforms further amplifies this phenomenon, causing an opinion overestimation. On the other, being such platforms a standard means of information, they may lead to an opinion underestimation among individuals with limited access to them.} This letter aims to address two primary challenges in the study of opinion dynamics: identifying the minimal conditions under which explosive polarization or depolarization occurs, and demonstrating how these phenomena {intertwine with the nature of such transitions}. To this aim, we will focus on \emph{opinion perception}, defined as the cognitive and social processes through which individuals interpret and integrate the opinions of others, influenced by biases, social context, and interaction dynamics \cite{milli2021opinion}. Drawing an analogy to nonlinear incidence in epidemiological models, which has shown that such terms can significantly alter the dynamics \cite{liu_influence_1986, liu_dynamical_1987}, we will model asymmetric opinion perception through similar nonlinear terms and demonstrate how this leads to explosive transitions. Asymmetric interactions have proven to be decisive in optimizing search strategies in congested networks \cite{carletti2020nonlinear}, enhancing prevention strategies in epidemic control \cite{siebert2022nonlinear}, and generating power-law segregation patterns in vegetation \cite{de2024emergence}, to mention a few.

To gain deeper insights into the dynamics of opinion spread, we will adopt a mean-field approach to develop an analytically tractable model that allows for a bifurcation analysis \cite{newman_networks_2018}. By identifying key bifurcation points and understanding the hysteresis effects, we lay the groundwork for developing strategies to manage and influence opinion propagation. Additionally, this method opens new pathways for understanding how social communities induce polarization in social networks. Continuous phase transitions, as opposed to explosive ones, allow for pattern prediction near criticality \cite{cross_pattern_2009}. Accordingly, we show that opinion clusters are shaped according to the underlying communities, giving rise to polarization. Furthermore, polarization {shifts} when the bifurcation switches from continuous to discontinuous, and even more so as the transition becomes more abrupt, {culminating in an opinion shift as} the new opinion {takes over}.

Although models based on principles such as homophily \cite{dandekar2013biased, baumann2020modeling} or bounded confidence \cite{deffuant2000mixing, Rainer2002-RAIODA} have been successful in emulating opinion polarization, they often fail to capture the abrupt transitions of opinions commonly observed in social networks. In contrast, Ref. \cite{ojer_modeling_2023} uses a phase-coupled model to explain explosive depolarization, analogous to the abrupt synchronization observed in the Kuramoto model \cite{kuramoto1975self}. Understanding the transition from consensus to polarization necessitates studying the nature of bifurcations, which in turn requires continuous variable models \cite{strogatz_nonlinear_2007}. Compartmental models such as the Susceptible-Infected-Susceptible (SIS) one have long been used as a mathematical framework to describe the spread of opinions within a population \cite{kiss_mathematics_2017}. The SIS model has been applied to study information spread on social networks, where individuals repeatedly adopt and abandon opinions based on interactions \cite{castellano_statistical_2009}. It has also been used to analyze rumor spreading, where individuals repeatedly believe or disbelieve a rumor \cite{nekovee_theory_2007}, and the spread of behaviors and social norms, where individuals switch between conforming and deviating based on peer influence and personal experience \cite{granovetter_threshold_1978, centola_spread_2010}.


Unlike disease contagion, in opinion dynamics, unaffected individuals in contact with those who have switched their opinions may change their minds more easily than expected or may be reluctant to do so. To facilitate the introduction of asymmetric perception, we will {model individuals} as metanodes, each containing a fixed number \(N\) of noninteractive particles (units) {that represent fractions of the individual's opinion}. These units are classified {either as \(S_i\), carrying the old opinion, or as \(I_j\), adopting the new opinion and replacing the corresponding \(S_j\)}. A susceptible particle \(S_i\) can become influenced by interacting with \(I_j\) at a adoption rate \(\beta\), modeled by \(S_i + I_j \xrightarrow{\beta} I_i + I_j\). In addition, an influenced unit \(I_i\) can lose interest in or forget the new opinion, transitioning from \(I_i \xrightarrow{\gamma} S_i\) at a revision rate \(\gamma\). To contrast the symmetry of contagion in the linear incidence model described above, we introduce an additional term that disrupts the 1-to-1 ratio between susceptible units \(S_i\) and influenced units \(I_j\). Specifically, we consider that a susceptible unit \(S_i\) can switch state only through interactions with multiple influenced units \(I_j\) at a rate \(\alpha\), represented as:
\begin{equation}
S_i + \underbrace{I_j + I_j + \cdots + I_j}_{d} \xrightarrow{\alpha} I_i + d I_j\,,
\label{eq:NL_incid}
\end{equation}
but the reverse scenario, where multiple susceptible units \(S_i\) adopt the new opinion simultaneously influenced by a single unit \(I_i\), can also occur. {Intuitively, this can be understood as a \(d\)-times repeated exposure to the same individual opinion or, conversely, as a \(d\)-times faster adoption.} To derive the evolution equation of node \(i\), let \(x_i(t)\) denote the fraction of \(n\) individuals at node \(i\) who have adopted the new opinion at time \(t\), i.e., in the thermodynamic limit \(x_i(t) = \lim_{N \to \infty} {nI_i}/{N}\). Similarly, \(s_i(t) = \lim_{N \to \infty} {(N-n)S_i}/{N}\), reflecting the total population normalization, i.e., \(x_i(t) + s_i(t) = 1\), represents the probability of node \(i\) who has not yet adopted the new opinion and is susceptible to being influenced by others who have adopted it. Using the mass action law, {the new opinion's probability evolution over a $\Omega$-nodes network is:}
\begin{equation}
\dot{x_i} = -\gamma x_i + \beta \left(1 - x_i\right) \sum_{j} A_{ij} x_j \left(1 + \alpha x_{j}^{d-1}\right)\,\forall i\,,
\label{eq:ME}
\end{equation}
where \(A_{ij}\) represents the adjacency matrix entries, with \(A_{ij} = 1\) indicating a connection between nodes \(i\) and \(j\), all diagonal elements \(A_{ii} = 0\), and for simplicity of representation \(\alpha \to \alpha\beta\). Here, {the term \(\beta \left(1+\alpha x_{j}^{d-1}\right)\) represents a modified density-dependent contagion rate, i.e., \(\hat{\beta}(x_j)\), where} \(d\) is the degree of the nonlinear incidence, {typically estimated from experimental data---a task that goes beyond the goal of this paper. In social dynamics, in particular, the parameter \(d\) can generalize to include non-integer values \cite{horn1972general},} indicating \emph{underestimated} perception for \(d > 1\) and \emph{overestimation} for \(0 < d < 1\).  Eq. \eqref{eq:ME} models the dynamics of opinion spread in a network, capturing the effects of both linear and nonlinear incidence interactions to reflect the complex influence of human perception. We will consider this dynamics throughout this paper to understand {the nature of transitions driving} opinion {polarization within} a social network.

\begin{figure*}
    \centering
        \includegraphics[width=\textwidth]{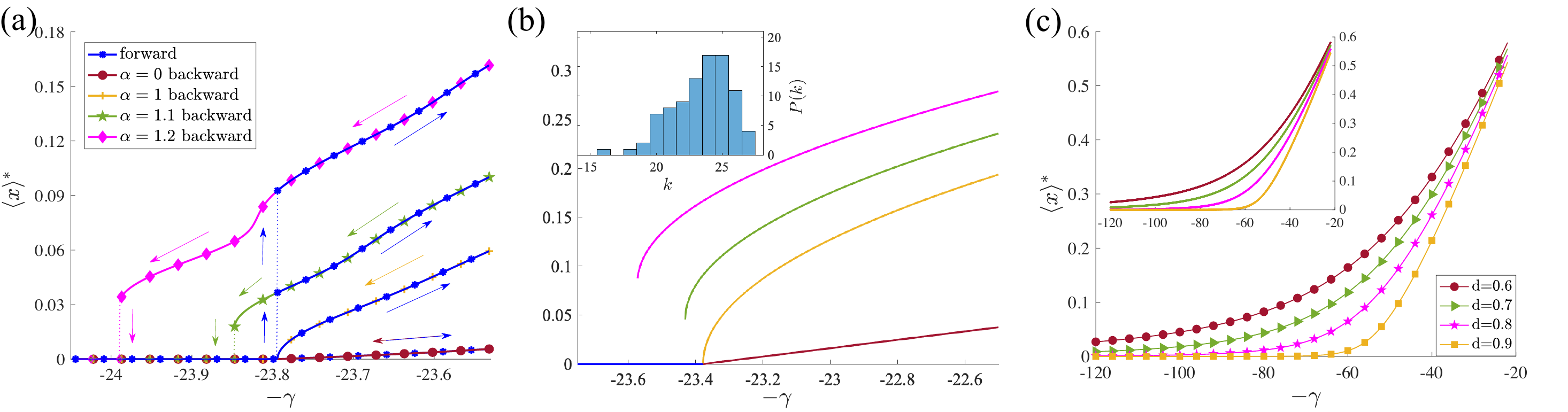}
    \caption{The mean equilibrium state, \(\langle x \rangle^*\), is plotted against the recovery rate \(-\gamma\). {Understimate (\(d=2\)):} (a) \textit{Numerical}: Forward direction (blue asterisks, blue arrows) and backward direction (other colors, arrows). For \(\alpha = 0\) (red circles), a transcritical transition occurs. For \(\alpha < 1\) (purple triangles for \(\alpha = 0.8\), yellow pluses for \(\alpha = 1\)), a supercritical bifurcation is observed. For \(\alpha > 1\) (green stars for \(\alpha = 1.1\), magenta diamonds for \(\alpha = 1.2\)), the transition is discontinuous (explosive), showing hysteresis and resistance to depolarization, leading to an abrupt no-opinion switch. (b) \textit{Analytical}: Equilibrium states from Eq. \eqref{eq:ME} confirm bistability in numerical results, matching the same \(\alpha\) values with corresponding colors. The inset shows the degree distribution. {(c) Overstimate (\(0 < d < 1\)): (Inset) Numerical equilibrium states of the average opinion and (inset) the analytical approximation with the same color legend. The network used is an unweighted modular one of 90 nodes, divided into 3 Erd\H{o}s--R\'enyi modules, each with a connection probability of \(p = 0.8\), with single edges linking the modules, and \(\beta = 1\).}}
    \label{fig:Fig1}
\end{figure*}


We can determine the critical point where the system globally adopts the new opinion by introducing a small perturbation \(\delta \mathbf{x}\) to the zero equilibrium state \( \mathbf{x}^{\ast} = \mathbf{0}\). Substituting \( \mathbf{x} = \mathbf{x}^{\ast} + \delta \mathbf{x} \) in Eq. \eqref{eq:ME} and after a straightforward linearization, we obtain \(\dot{\delta \mathbf{x}} = (-\gamma \mathbb{I} + \beta \mathbf{A}) \delta \mathbf{x}\), where \( \mathbf{J} = -\gamma \mathbb{I} + \beta \mathbf{A} \) is the Jacobian matrix. The spectrum of \(\mathbf{A}\) is shifted left by \(\gamma\), indicating that stability is determined by the largest eigenvalue of the adjacency matrix, \(\lambda_A^{\text{max}}\), with \(\lambda_J = -\gamma + \beta \lambda_A\) where all eigenvalues are real due to the network's undirectedness. The stability of the system depends solely on the contagion and recovery rates, and is independent of the nonlinear incidence. 
{The} new opinion spreads if the largest eigenmode of the network{, scaled with the contagion rate \(\beta\), exceeds} the recovery rate {(\(\beta\lambda_A^{\text{max}} > \gamma\))}, and dies down otherwise, a feature characterizing the SIS model, \cite{newman_networks_2018, kiss_mathematics_2017}.

In Fig. \ref{fig:Fig1} (a), we use the reversion rate \(\gamma\) as the control parameter and record the equilibrium state of the system described by Eq. \eqref{eq:ME} as the average value of the opinion density across the nodes \(\langle x \rangle = \sum_i x_i / \Omega\) for different values of \(\alpha\). {To allow the new and old opinions to polarize, we consider a network with communities, represented, without loss of generality, by a toy model of a single-linked, three-module connected graph.} As predicted, the first critical point is reached when \(\gamma\) equals \(\lambda_A^{\text{max}} \approx 23.8\) for all the scenarios considered {(with \(\beta=1\))}. The most notable feature, however, is that in the absence of nonlinear incidence (\(\alpha = 0\)) and when \(\alpha \leq 1\), the equilibrium curves are continuous and appear respectively linear and parabolic, characteristic of transcritical and supercritical pitchfork bifurcations. When \(\alpha > 1\), a first-order (explosive) phase transition emerges, with the discontinuity jump increasing with \(\alpha\). Conversely, when the control parameter \(-\gamma\) decreases, the backward trajectory manifests hysteresis, maintaining the {new opinion before reaching a second critical point, where it abruptly reverts to the old one}. The bistability observed here aligns with the population's resilience to {revisiting the old opinion}. Furthermore, when such depolarization occurs, it results in an instant switch back to the {original} state.

To understand the mechanisms underlying explosive polarization/depolarization, we will reduce the system \eqref{eq:ME} to a one-dimensional equation by considering the evolution of the average opinion.
To decouple the correlated terms, we make the following considerations: first, we assume that highly connected nodes are more influenced by their neighbors' opinions, as follows: \(x_i(t) \sim k_i\). Although this assumption underpins the well-known degree-based mean-field (DBMF) approximation \cite{newman_networks_2018, pastor2001epidemic_1, pastor2001epidemic_2, barthelemy2004velocity, barthelemy2005dynamical}, it simplifies the system's dimensionality to the maximum degree range, thus limiting analytical exploration. To facilitate analytical progress, we introduce a second assumption: a narrow {(i.e., small \(\left(k_{\text{max}} - k_{\text{min}}\right)/\langle k \rangle\))} and symmetric degree distribution \(P(k)\). 
Based on such considerations we express \(x_i = \langle x \rangle + \delta x_i\) and \(k_i=\langle k \rangle + \delta k_i\) in terms of deviations \(\delta x_i\) and \(\delta k_i\) from their respective means to obtain:
\begin{align}
\langle \dot{x} \rangle = \left(\tilde{\beta}-\gamma\right) \langle x \rangle - \tilde{\beta} \langle x \rangle^2 + \alpha \tilde{\beta} \left(1 - \langle x \rangle\right) \langle x \rangle^d
\label{eq:mean_x}
\end{align}
where \(\tilde{\beta}=\beta\langle k \rangle\) and we have used the binomial approximation \(\left(\langle x \rangle + \delta x_j\right)^d \approx \langle x \rangle^d + d \langle x \rangle^{d-1} \delta x_j\). We have eliminated terms related to \(\sum_i \delta x_i=0\) and \(\sum_i \delta k_i=0\) following our assumptions. We have also neglected higher order terms \(\delta x_i \delta k_i\), as deviation terms \(\delta x_i\) and \(\delta k_i\) are assumed to be smaller than their respective means, i.e., \(\delta x_i \ll \langle x \rangle\) and \(\delta k_i \ll \langle k \rangle\). For details of the derivation refer to the Supplemental Material (SM). 

\begin{figure*}
    \includegraphics[width=\linewidth]{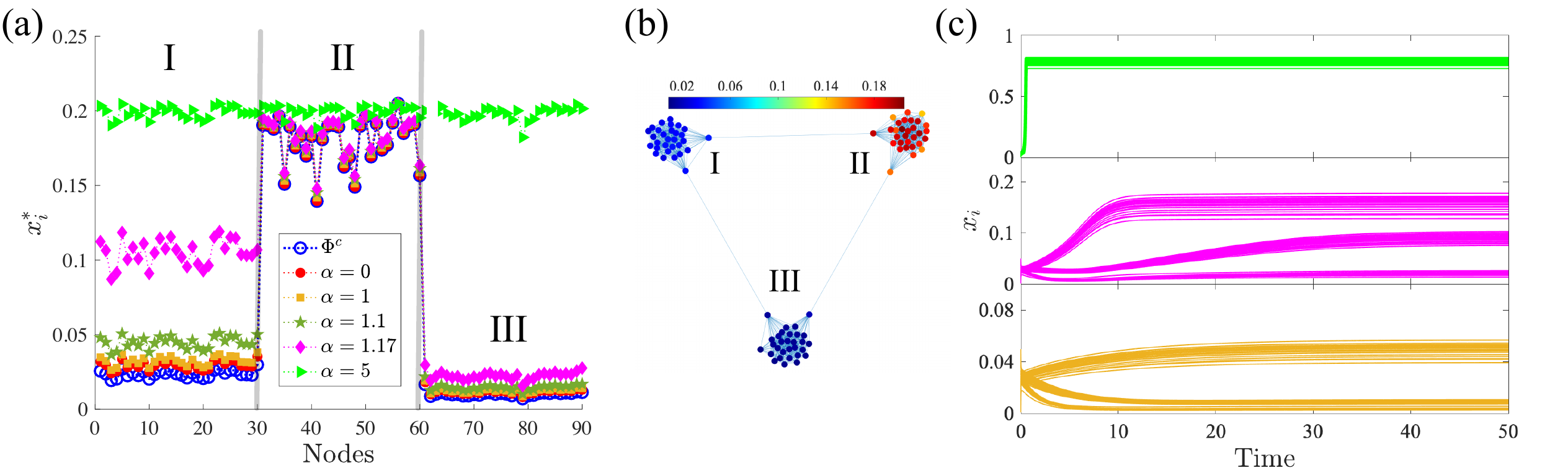}
    \caption{{(a)} The equilibrium state for each node \(x_i^*\), plotted for various values of \(\alpha\) is {scaled for comparison} to the critical eigenvector {, where \(-\gamma=23.734\)}. In the continuous phase transition (\(\alpha \leq 1\)), the pattern is similar to the critical eigenvector. The pattern starts to differ for \(\alpha = 1.1\) when the transition becomes discontinuous and {for \(\alpha = 2\), an opinion shift occurs as the discontinuity escalates}. {(b)} The {network} pattern for \(\alpha=1\){, with the colorbar for the equilibrium \(x_i^*\)}. {(c) Opinion polarization evolution for values of the \(\alpha\) parameter with the same color code.} The {other} parameters and network match those in Fig. \ref{fig:Fig1}.}
\label{fig:Fig2}
\end{figure*}

Fig. \ref{fig:Fig1} (b) demonstrates that Eq.~\eqref{eq:mean_x} exhibits qualitatively similar behavior to that observed earlier for the full system \eqref{eq:ME}, thereby confirming the ansatz introduced earlier when applied to a symmetric narrow degree distribution (shown by the histogram in the inset). In particular, {considering \(d=2\) simplifies} Eq. \eqref{eq:mean_x} to
\[
\langle \dot{x} \rangle = \left(\tilde{\beta} - \gamma\right) \langle x \rangle + \tilde{\beta}(\alpha - 1) \langle x \rangle^2 - \alpha \tilde{\beta} \langle x \rangle^3
\]
allowing a direct comparison between the linear incidence terms and the nonlinear ones. Notably, Fig. \ref{fig:Fig1}(b) reveals a different critical point where the linearization of the one-dimensional reduced system \eqref{eq:mean_x} yields \(\gamma_c = \tilde{\beta}=\beta \langle k\rangle\), in contrast to the previous result of \(\gamma_c = \beta\lambda_{A}^{\max}\). This shift in criticality is due to the mean degree being bounded above by the Perron eigenvalue \(\langle k \rangle \leq \lambda_A^{\max}\), which, as proved in SM, is a straightforward consequence of the Rayleigh quotient \cite{horn2012matrix}.
Specifically, one can observe that when \(\alpha = 0\) and \(\alpha = 1\), the supercritical (pitchfork) and transcritical normal forms are respectively obtained \cite{strogatz_nonlinear_2007}, corroborating the numerical results. Similarly, when \(0 < \alpha < 1\), the negative cubic term of the supercritical bifurcation is further reinforced by the square term. The results differ when \(\alpha > 1\): the cubic term is counterbalanced by the positivity of the quadratic term, changing the nature of the bifurcation to subcritical and leading to explosive polarization. In this case, bistability is ensured, and consequently, hysteresis in the backward path, which results in explosive depolarization. Finally, in the SM, it is shown that the nonlinear incidence terms alone are responsible for the discontinuous phase transitions.
The dynamics change significantly when considering the overestimated case \(0 < d < 1\), as illustrated in {Fig. \ref{fig:Fig1} (c)} for various values of the perception parameter \(d\), shown for both the full system \eqref{eq:ME} (main) and the analytical proxy \eqref{eq:mean_x} (inset). In this case, we focus exclusively on the state where the new opinion is present \(\langle x \rangle \neq 0\) \footnote{We have numerically shown that the fixed point \(x_i = 0 \, \forall i\) is unstable for Eq. \eqref{eq:ME}. For the analytical proxy, since Eq. \eqref{eq:mean_x} for \(0 < d < 1\) is not analytical at the origin, one can use \(V(\langle x \rangle) = \frac{1}{2} \langle x \rangle^2\) as a candidate Lyapunov function to analyze the stability of the fixed point \(\langle x \rangle = 0\). This function is positive definite. Evaluating \(\dot{V}(\langle x \rangle) = \langle x \rangle \langle \dot{x} \rangle\) gives \(\dot{V}(\langle x \rangle) = (\tilde{\beta} - \gamma) \langle x \rangle^2 - \tilde{\beta} \langle x \rangle^3 + \alpha \tilde{\beta} \langle x \rangle^{d+1} - \alpha \tilde{\beta} \langle x \rangle^{d+2}\). For \(\langle x \rangle = \epsilon > 0\) with \(\epsilon \ll 1\), \(\dot{V}(\epsilon) \approx (\tilde{\beta} - \gamma)\epsilon^2\). Hence, \(\dot{V}(\epsilon) < 0\) if \(\tilde{\beta} < \gamma\), indicating stability. This {shows a gap} between {mean-field and individual-based approaches}.}. When the new opinion is overestimated, individuals will always embrace it. However, the average amount of the new opinion decreases asymptotically with the reversion rate \(\gamma\), meaning that although some of it is retained, it is largely compensated by the population's capability of forgetting or losing interest. Interestingly, and in contrast with the underestimating case, overestimation mitigates the abrupt depolarization.

So far, we have focused on the role that asymmetric perception plays in explosive transitions in {the new} opinion states. {In the SIS model, polarization manifests as clusters of connected individuals with the new opinion \(I\) or the complementary old opinion \(S\).} Understanding the nature of transitions {between states} is crucial for understanding polarization, {as} continuous bifurcations not only prevent abrupt transitions to new opinions but also allow for predicting differential effects on individuals \cite{cross_pattern_2009}. By expanding the perturbation across the basis of adjacency matrix eigenvectors \(\delta \mathbf{x} = \sum_\nu c_\nu e^{\lambda_J^{\nu}} \boldsymbol{\Phi}^\nu\), one can observe that, in the linear regime, the system follows the trajectory shaped by the linear combination of the critical eigenvectors, those corresponding to \(\beta\lambda_A^{\nu} > \gamma \). Pattern formation theory shows that near the criticality of a continuous phase transition, the pattern indicated by the linear solution will persist in the nonlinear equilibrium state as well \cite{cross_pattern_2009}. Following this logical line, {for continuous tranistions} the shape of the pattern will resemble that of the critical eigenvectors, as long as the instability conditions {hold near the critical point.}

Why, then, can the opinion pattern polarize? Through spectral perturbation techniques (see SM), it can be shown that for an \(M\)-modular network, the eigenvectors corresponding to the \(M\) largest eigenvalues have near-zero entries across all modules except one. This implies that if such eigenvectors are critical, the final pattern will localize within only one of the network's communities \cite{siebert2020role, asllani2022symmetry}. To illustrate such a phenomenon, in {Fig. \ref{fig:Fig2}}, we consider a 3-modules network where we ensure that only the principal eigenvector \(\boldsymbol{\Phi}_A^{\text{max}}\) is critical. This choice is two-fold: first, it ensures proximity to the transition threshold, and second, according to the Perron-Frobenius theorem, only the principal eigenvector has all positive entries, a logical requirement of our model. For values of the parameter \(\alpha \leq 1\), where the transition remains continuous, the pattern closely mimics the shape of the critical eigenvector. However, as the nonlinear incidence terms dominate (\(\alpha = 1.1\)), the transition becomes discontinuous, and the green star pattern begins to deviate from the critical eigenvector. A {significant} increase in the nonlinear terms to {\(\alpha = 2\) (green triangles) not only} accentuates this deviation further{, but transforms the new opinion into a uniform pattern among individuals, resulting in an explosive opinion shift}. Notice that richer polarization dynamics is in principle possible if further eigenvectors are considered to be critical {(see SM)}.


In this letter, we proposed a minimal yet robust opinion dynamics model to explain explosive opinion {spreading and polarization}. To allow opinions to polarize, we utilize a modular network known for clustering patterns according to the underlying nodes' communities. Our primary hypothesis is that the {switching to} a new opinion, unlike disease contagion, is based on asymmetric perception, which can be modeled through nonlinear incidence. Numerical simulations verify this hypothesis, showing a transition from continuous to discontinuous as nonlinear incidence terms dominate. To understand the conditions leading to explosive polarization/depolarization, we develop a mean-field approximation based on the observation that individuals are influenced by a probability proportional to their number of connections. Assuming a symmetric and narrow degree distribution, this approximation allows us to derive a one-dimensional normal form to analytically describe the bifurcations, confirming the numerical results. This suggests that to mitigate or prevent abrupt opinion shifts, individuals' perceptions should be as proportional as possible. These findings have practical implications for marketing, political campaigns, and public health communications, where controlling opinion dynamics is crucial.

\bibliographystyle{apsrev4-2}
\bibliography{citations}


\clearpage

\onecolumngrid

\begin{center}
    {\Large \textbf{Supplemental Material}}\\[.5em]
    {\large \textbf{Explosive opinion spreading with polarization and depolarization via asymmetric perception}}
\end{center}

\begin{center}
    Haoyang Qian and Malbor Asllani \\
    \textit{Department of Mathematics, Florida State University,\\
    1017 Academic Way, Tallahassee, FL 32306, United States of America}
\end{center}

\section{Mean-field proxy for the SIS model}

The individual-based model of Eq. (1) in the main text is given as:

\begin{equation}
\dot{x_i} = -\gamma x_i + \beta (1-x_i) \sum_{{j}} A_{i{j}} x_{{j}} (1+\alpha x_{{j}}^{d-1}),\, \forall i
\label{eq:1}
\end{equation}
where \(\alpha\) is the coefficient quantifying the nonlinear incidence, and \(d\) is the degree of the nonlinear incidence.\\

Assumptions:\\

\begin{enumerate}
  \item \(x_i(t) \sim k_i\), where \(k_i\) represents the degree of node \(i\).
  \item The degree distribution \(P(k)\) is symmetric and narrow.\\
\end{enumerate}

Given the equation:

\begin{equation}
\langle \dot{x} \rangle = -\gamma \langle x \rangle + \frac{\beta}{N} \sum_{j} x_j k_j - \frac{\beta}{N} \sum_{i,j} A_{ij} x_i x_j + \alpha \frac{\beta}{N} \sum_{i} (1 - x_i) \sum_{j} A_{ij} x_j^d
\label{eq:first}
\end{equation}

Substituting \( x_i = \langle x \rangle + \delta x_i \) and \( k_i = \langle k \rangle + \delta k_i \):

\[
\langle \dot{x} \rangle = -\gamma \langle x \rangle + \frac{\beta}{N} \sum_{j} (\langle x \rangle + \delta x_j)(\langle k \rangle + \delta k_j) - \frac{\beta}{N} \sum_{i,j} A_{ij} (\langle x \rangle + \delta x_i)(\langle x \rangle + \delta x_j) + \alpha \frac{\beta}{N} \sum_{i} \left(1 - (\langle x \rangle + \delta x_i)\right) \sum_{j} A_{ij} \left(\langle x \rangle + \delta x_j\right)^d
\]

Expanding and simplifying:

\begin{align*}
\langle \dot{x} \rangle \approx -\gamma \langle x \rangle + \frac{\beta}{N} \sum_{j} \big(\langle x \rangle \langle k \rangle &+ \langle x \rangle \delta k_j+ \delta x_j \langle k \rangle + \delta x_j \delta k_j\big) - \frac{\beta}{N} \sum_{i,j} A_{ij} \left(\langle x \rangle^2 + \langle x \rangle \delta x_j + \delta x_i \langle x \rangle + \delta x_i \delta x_j\right) +\\ &+ \alpha \frac{\beta}{N} \sum_{i} \left(1 - \langle x \rangle - \delta x_i\right) \sum_{j} A_{ij} \left(\langle x \rangle^d + d \langle x \rangle^{d-1} \delta x_j\right)\,,
\end{align*}

{where we have approximated \((x+\delta x)^d \approx x^d + dx^{d-1}\delta x\) for \(\delta x \ll x\).} Combining the terms, we obtain:

\begin{align}
\langle \dot{x} \rangle &= -\gamma \langle x \rangle + \beta \langle x \rangle \langle k \rangle - \beta \langle x \rangle^2 \langle k \rangle + \alpha \beta \left(1 - \langle x \rangle\right) \langle k \rangle \langle x \rangle^d
+ \frac{\beta}{N} \sum_{j} \left(\langle x \rangle \delta k_j + \delta x_j \langle k \rangle + \delta x_j \delta k_j\right)+\nonumber\\ &- \frac{\beta}{N} \sum_{i,j} A_{ij} (\langle x \rangle \delta x_j + \delta x_i \langle x \rangle + \delta x_i \delta x_j) + \alpha \frac{\beta}{N} {\cancel{\sum_{i} \delta x_i}} \sum_{j} A_{ij} \langle x \rangle^d + \alpha \frac{\beta}{N} (1 - \langle x \rangle) \sum_{{i,}j} A_{ij} d \langle x \rangle^{d-1} \delta x_j +\nonumber\\ &- \alpha \frac{\beta}{N} \sum_{{i,}j} A_{ij} d \langle x \rangle^{d-1} {\delta x_i }\delta x_j\nonumber  
\end{align}

{
\begin{align}
&= -\gamma \langle x \rangle + \beta \langle x \rangle \langle k \rangle - \beta \langle x \rangle^2 \langle k \rangle + \alpha \beta \left(1 - \langle x \rangle\right) \langle k \rangle \langle x \rangle^d
+ \frac{\beta}{N} \left[ \langle x \rangle\cancel{\sum_{j} \delta k_j} + \cancel{\langle k \rangle \sum_{j} \delta x_j}  + \sum_{j} \delta x_j \delta k_j \right]+\\ &- \frac{\beta}{N}\left[ 2\langle x \rangle \sum_j k_j \delta x_j + \sum_{i,j} A_{ij}\delta x_i \delta x_j\right] + \alpha \frac{\beta}{N}d \langle x \rangle^{d-1} \left[ (1 - \langle x \rangle) \sum_j k_j \delta x_j - \sum_{i,j} A_{ij} \delta x_i \delta x_j\right]\nonumber\\ 
&= -\gamma \langle x \rangle + \beta \langle x \rangle \langle k \rangle - \beta \langle x \rangle^2 \langle k \rangle + \alpha \beta \left(1 - \langle x \rangle\right) \langle k \rangle \langle x \rangle^d + \langle k \rangle \frac{\beta}{N} \left(\alpha d \langle x \rangle^{d-1} (1 - \langle x \rangle)- 2\langle x \rangle \right)\cancel{\sum_j \delta x_j} \nonumber \\ &+ \frac{\beta}{N} \left[\left(1-2\langle x \rangle +\alpha d \langle x \rangle^{d-1} \left(1 - \langle x \rangle\right) \right)\sum_{j} \delta x_j \delta k_j - \left(1+\alpha d \langle x \rangle^{d-1}\right) \sum_{i,j} A_{ij}\delta x_i \delta x_j\right] \nonumber 
\label{eq:expansion}
\end{align}
}

Given the symmetric and narrow degree distribution, higher-order terms and products of deviations (such as \(\delta x_j \delta k_j\) and \(\delta x_i \delta x_j\)) are small because \(\delta k_i \ll \langle k \rangle\) and \(\delta x_i \ll \langle x \rangle\). Additionally, the symmetry of the degree distribution \(P(k)\) implies that the sums of \(\delta k_i\) and \(\delta x_i\) average out to zero. {Furthermore, the term \(\sum_{i,j} A_{ij}\delta x_i \delta x_j\) consists of both positive and negative contributions, resulting in a total value that remains relatively small.} Consequently, the terms involving deviations \(\delta k_i\) and \(\delta x_i\) can be neglected. {The result of Eqs. (6) provides insight into the difference in the critical points between the individual-based definitions and the mean-field approximation. When all terms \(\langle x \rangle\) are considered, this difference depends on \(\sum_j \delta x_j \delta k_j\).} Combining these considerations, the final expression for \(\langle \dot{x} \rangle\) is:

\begin{equation}
\langle \dot{x} \rangle = (\tilde{\beta}-\gamma) \langle x \rangle - \tilde{\beta} \langle x \rangle^2 + \alpha \tilde{\beta} (1 - \langle x \rangle) \langle x \rangle^d
\label{eq:2}
\end{equation}
where we have used \(\tilde{\beta} = \beta \langle k \rangle\).

\subsection{Analysis for different \(\alpha\) values}

Starting from the previous Eq. \eqref{eq:2}, we consider the following cases:

\subsubsection{Case 1: \(\alpha = 0\)}

When \(\alpha = 0\), the equation simplifies to:

\begin{equation}
\langle \dot{x} \rangle = (\tilde{\beta}-\gamma) \langle x \rangle - \tilde{\beta} \langle x \rangle^2
\label{eq:3}
\end{equation}

This is a quadratic equation, corresponding to the normal form of a transcritical bifurcation \cite{strogatz_nonlinear_2007}.

\subsubsection{Case 2: \(\alpha > 0\) and \(d = 2\)}

When \(\alpha > 0\) and \(d = 2\), the equation becomes:

\[
\langle \dot{x} \rangle = -\gamma \langle x \rangle + \tilde{\beta} \langle x \rangle - \tilde{\beta} \langle x \rangle^2 + \alpha \tilde{\beta} (1 - \langle x \rangle) \langle x \rangle^2
\]

Factoring out \(\tilde{\beta}\) only for the middle term, we get:

\begin{equation}
\langle \dot{x} \rangle =  (\tilde{\beta} - \gamma) \langle x \rangle + \tilde{\beta}(\alpha - 1) \langle x \rangle^2 - \alpha \tilde{\beta} \langle x \rangle^3
\label{eq:4}
\end{equation}

This is a cubic equation, corresponding to the normal form of a supercritical pitchfork bifurcation for \(\alpha \leq 1\) and a subcritical pitchfork bifurcation for \(\alpha > 1\) \cite{strogatz_nonlinear_2007}.

\subsection{A purely nonlinear incidence model}

Eq. \eqref{eq:2} incorporates both linear and non-linear terms, a choice made purely for comparison purposes. However, since our primary interest lies in the direct effect that the non-linear terms alone have on the explosive transitions, we will simplify the previous equation to:

\begin{equation}
\dot{x_i} = -\gamma x_i + \beta \alpha (1 - x_i) \sum_{j} A_{ij} x_j^d,\, \forall i
\label{eq:5}
\end{equation}
where in the limit \(d = 1\), we recover the linear incidence term.

For reasons we will explain in the following, we need to slightly modify the above equation to include another slack compartment. Unlike the setting discussed throughout the paper, now within an individual, there exist three different possible states: Susceptible (S), whose opinions are subject to change, Infected (I), who propagate new opinions, and a subset of units inside an individual (D), who remain indifferent to new opinions; they hold steadfast to their own beliefs and cannot be influenced. This means that the normalization condition now extends to \( S + I + D = 1 \). For the sake of representation, we will simplify the notation by substituting \( D \) with \( 1 - D \). Thus, we have the relationship \( S + I = D \) with \( 0 \leq D \leq 1 \). Consequently, Eq. \eqref{eq:5} can be modified as follows:

\begin{equation}
\dot{x_i} = -\gamma x_i + \beta \alpha (D - x_i) \sum_{j} A_{ij} x_j^d,\, \forall i\,.
\label{eq:6}
\end{equation}

A straightforward application of the mean-field approximation method would give us Eq. \eqref{eq:4} again, where we have dropped the linear incidence term and weighted the square term with \(D\). For the particular case \(d = 2\), it can be expressed as follows:

\begin{equation}
\langle \dot{x} \rangle = -\gamma \langle x \rangle + D \alpha \tilde{\beta} \langle x \rangle^2 - \alpha \tilde{\beta} \langle x \rangle^3
\label{eq:7}
\end{equation}

From the polynomial form of Eq. \eqref{eq:7}, if \(-\gamma < 0\), the system is stable, which means new opinions cannot spread. If \(-\gamma > 0\), the system becomes unstable, meaning that new opinions will spread out. In Fig. \ref{fig:figure17}, we have shown the bifurcation curves of the numerical Eq. \eqref{eq:6} and the analytical proxy Eq. \eqref{eq:7} for different values of \(D\). The numerical results closely align with the analytical predictions, indicating a high degree of similarity between the two.
In particular, it can be noted that for \(D = 1\), and when \(-\gamma = 0\), the equilibrium \(\langle x \rangle^*\) stabilizes at 1 and cannot increase further. This phenomenon occurs because, as indicated by Eq. \eqref{eq:7}, when \(-\gamma = 0\), there are two equilibrium points: \(\langle x \rangle^* = 0\) and \(\langle x \rangle^* = 1\). In fact, when looking for a non-null fixed point of Eq. \eqref{eq:7}, we need to find the roots of \(\gamma + \alpha \tilde{\beta} \langle x \rangle(D-\langle x\rangle)=0\), and if \(D=1\), it has no solution for \(-\gamma > 0\). This is also the reason why we introduced the compartment of indifferent individuals \(D\) to bypass this technical limitation. Unlike the model considered in the main text, here, in all the scenarios, the system exhibits a discontinuous transition for both the polarization and depolarization regimes.

\begin{figure}[h!]
    \centering
\includegraphics[width=\textwidth]{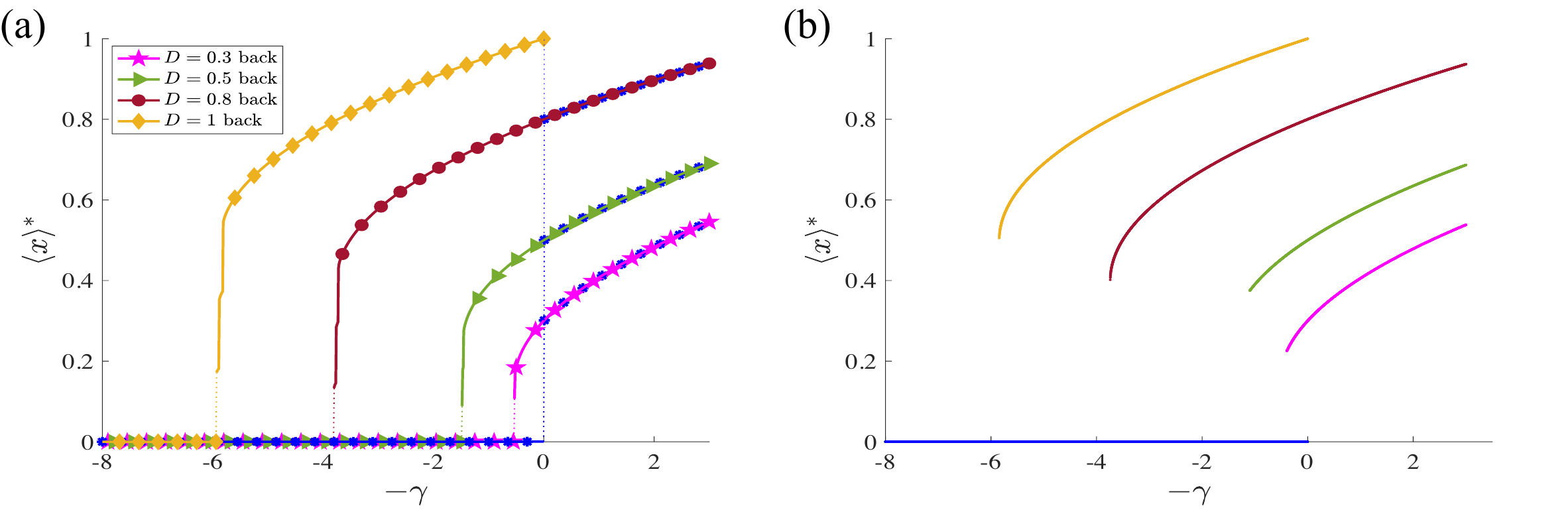}
    \caption{The mean equilibrium state, \(\langle x \rangle^*\), is plotted against the recovery rate \(-\gamma\) with model parameters \(\beta = 1\) and \(d = 2\) for both numerical (a) and analytical (b) models. In the numerical case (left), the forward direction is shown with blue asterisks, and the backward direction is shown with other colors. The analytical results (right) use the same color scheme and parameters as the numerical case, corroborating the numerical findings.}
    \label{fig:figure17}
\end{figure}


\section{Spectral Properties of modular networks}

Here, we elucidate the characteristics of the principal eigenvector of the adjacency matrix by first recalling a part of the Perron–Frobenius theorem \cite{horn2012matrix}.

\begin{theorem}[Perron–Frobenius]
For {an irreducible non-negative} adjacency matrix $\bm{A}$, there exists a unique largest eigenvalue $\lambda_A^{\textnormal{max}}$ associated with a principal {(right)} eigenvector $\bm{\Phi}^{\textnormal{max}}$ whose components are all positive. Moreover, all other {(right)} eigenvectors possess at least one negative component.
\label{th:PF}
\end{theorem}

Let us denote \(\bm{A}_0\) as the initial adjacency matrix, with blocks \(\bm{A}_1, \bm{A}_2, \ldots, \bm{A}_M\) representing the adjacency matrices for \(M\) disconnected clusters of nodes, also known as modules \cite{newman_networks_2018}. The zero matrices \(\mathbf{0}\) (which may not necessarily be square and can be of different dimensions) indicate no connections between modules.

\[
\bm{A}_0 = \begin{bmatrix}
    \bm{A}_1 & \mathbf{0} & \cdots & \mathbf{0} \\
    \mathbf{0} & \bm{A}_2 & \cdots & \mathbf{0} \\
    \vdots & \vdots & \ddots & \vdots \\
    \mathbf{0} & \mathbf{0} & \cdots & \bm{A}_M
\end{bmatrix}
\]

To obtain a fully connected, yet strongly modular network with adjacency matrix \(\bm{A}\), the previous zero matrices will be replaced by sparse matrices, constituting a block matrix with a zero block diagonal denoted by $\bm{B_{\epsilon}}$, indicating a few connections between clusters.

\[
\bm{A} = \bm{A}_0 + \bm{B}_{\epsilon}
\]

The weakly perturbed matrix \(\bm{A}\) has slightly modified eigenvalues and eigenvectors compared to the original matrix \(\bm{A_0}\). The set of all eigenvalues (spectrum) of \(\bm{A}\) can be represented as:

\[
\sigma(\bm{A}) = \sigma(\bm{A_0}) + \beta_{\epsilon}
\]
where \(\sigma(\bm{A})\) and \(\sigma(\bm{A_0})\) are row vectors of the eigenvalues of \(\bm{A}\) and \(\bm{A_0}\), respectively, and \(\beta_{\epsilon}\) is a row vector of small perturbations.
Similarly, the eigenvectors of \(\bm{A}\) can be expressed as:

\[
\Phi(\bm{A}) = \Phi(\bm{A_0}) + \gamma_{\epsilon}
\]
where \(\Phi(\bm{A})\) and \(\Phi(\bm{A_0})\) are matrices whose columns are the eigenvectors of \(\bm{A}\) and \(\bm{A_0}\), respectively, and \(\gamma_{\epsilon}\) is a matrix containing small perturbations to the eigenvectors.
Thus, \(\beta_{\epsilon}\) and \(\gamma_{\epsilon}\) represent small changes to the eigenvalues and eigenvectors, respectively, due to the weak perturbation applied to the original matrix.
Then, for a given eigenvalue \(\lambda_0\) of the unperturbed matrix:

\[
\bm{A_0} \bm{\Phi}_0 = \begin{bmatrix}
    \bm{A}_1\bm{\Phi}_1 \\
    \bm{A}_2\bm{\Phi}_2 \\
    \vdots \\
    \bm{A}_M\bm{\Phi}_M
\end{bmatrix} = \lambda_0 \bm{\Phi}_0 = \begin{bmatrix}
    \lambda_0 \bm{\Phi}_1 \\
    \lambda_0 \bm{\Phi}_2 \\
    \vdots \\
    \lambda_0 \bm{\Phi}_M
\end{bmatrix}
\]
where 

\[
\bm{\Phi}_0 = \begin{bmatrix}
    \bm{\Phi}_1 \\
    \bm{\Phi}_2 \\
    \vdots \\
    \bm{\Phi}_M
\end{bmatrix}
\]
is the corresponding eigenvector of \(\lambda_0\), and \(\bm{\Phi}_1\), \(\bm{\Phi}_2\), \(\ldots\), \(\bm{\Phi}_M\) are the respective eigenvectors of \(\bm{A}_1\), \(\bm{A}_2\), \(\ldots\), \(\bm{A}_M\). When the matrix is weakly perturbed, the perturbed eigenvalue \(\lambda\) and the corresponding perturbed eigenvector \(\bm{\Phi}\) can be expressed as:

\[
\lambda = \lambda_0 + \beta_{\epsilon}^{(0)}\,,
\hspace*{2cm}
\bm{\Phi} = \bm{\Phi}_0 + \gamma_{\epsilon}^{(0)}
\]
where \(\beta_{\epsilon}^{(0)}\) is a small perturbation added to the eigenvalue, and \(\gamma_{\epsilon}^{(0)}\) is a vector containing small perturbations added to the eigenvector components.

Now, given that \(\bm{A}_0\) is a block matrix:

\[
\det(\bm{A}_0-\lambda \bm{I}) = \det(\bm{A}_1-\lambda \bm{I}) \det(\bm{A}_2-\lambda \bm{I}) \cdots \det(\bm{A}_M-\lambda \bm{I}) = 0
\]
this means that \(\sigma(\bm{A}_0) = \sigma(\bm{A}_1) \cup \sigma(\bm{A}_2) \cup \dots \cup \sigma(\bm{A}_M)\).

On the other hand, the eigenvalue \(\lambda_0\) should be an eigenvalue of all the blocks:

\[
\bm{A}_1 \bm{\Phi}_1 = \lambda_0 \bm{\Phi}_1, \quad \bm{A}_2 \bm{\Phi}_2 = \lambda_0 \bm{\Phi}_2, \quad \dots, \quad \bm{A}_M \bm{\Phi}_M = \lambda_0 \bm{\Phi}_M
\]

One way this can occur is if \(\lambda_0\) is an eigenvalue of \(\bm{A}_1\) only, meaning \(\lambda_0 \in \sigma(\bm{A}_1)\) and \(\lambda_0 \notin \sigma(\bm{A}_2) \cup \dots \cup \sigma(\bm{A}_M)\). In this case, the corresponding eigenvector is:

\[
\Phi(\bm{A}_0) = \begin{bmatrix}
    \Phi_1 \\
    \bm{0} \\
    \vdots \\
    \bm{0}
\end{bmatrix}.
\]
Similarly, if \(\lambda_0 \in \sigma(\bm{A}_2)\) and assuming \(\lambda_0 \notin \sigma(\bm{A}_1) \cup \dots \cup \sigma(\bm{A}_M)\), then

\[
{\Phi}(\bm{A}_0) = \begin{bmatrix}
    \bm{0} \\
    \bm{\Phi}_2 \\
    \vdots \\
    \bm{0}
\end{bmatrix}.
\]
And so on, for \(i \in \{1, 2, \dots, M\}\),

\[
{\Phi}(\bm{A}_0) = \begin{bmatrix}
    \bm{0} \\
    \vdots \\
    \bm{\Phi}_i \\
    \vdots \\
    \bm{0}
\end{bmatrix}.
\]
This implies that the eigenvectors of \(\bm{A}_0\) have one non-zero component corresponding to an eigenvector of one block \(\bm{A}_i\), while the others are zero. Given the uniqueness of eigenvectors, this is the only solution. Since \(\bm{A}\) slightly differs from \(\bm{A}_0\), the other components of the eigenvectors of \(\bm{A}\) are not exactly zero but close to it. {If the graph is connected, the adjacency matrix \(\bm{A}\) is irreducible, and the Perron-Frobenius Theorem \ref{th:PF} guarantees that the principal eigenvector corresponding to the largest eigenvalue is strictly positive.} In Fig. \ref{fig:figure7}, we compare the principal eigenvectors of \(\bm{A}\) and \(\bm{A}_0\). Two of the three blocks in \(\bm{A}_0\) are exactly zero, so their corresponding eigenvector components are zero. However, the eigenvector of \(\bm{A}\) shows small non-zero values in these blocks due to perturbations, indicating a slight spread of influence across all blocks, with most still concentrated in the dominant block.

\begin{figure}[h!]
\centering
\includegraphics[width=0.7\linewidth]{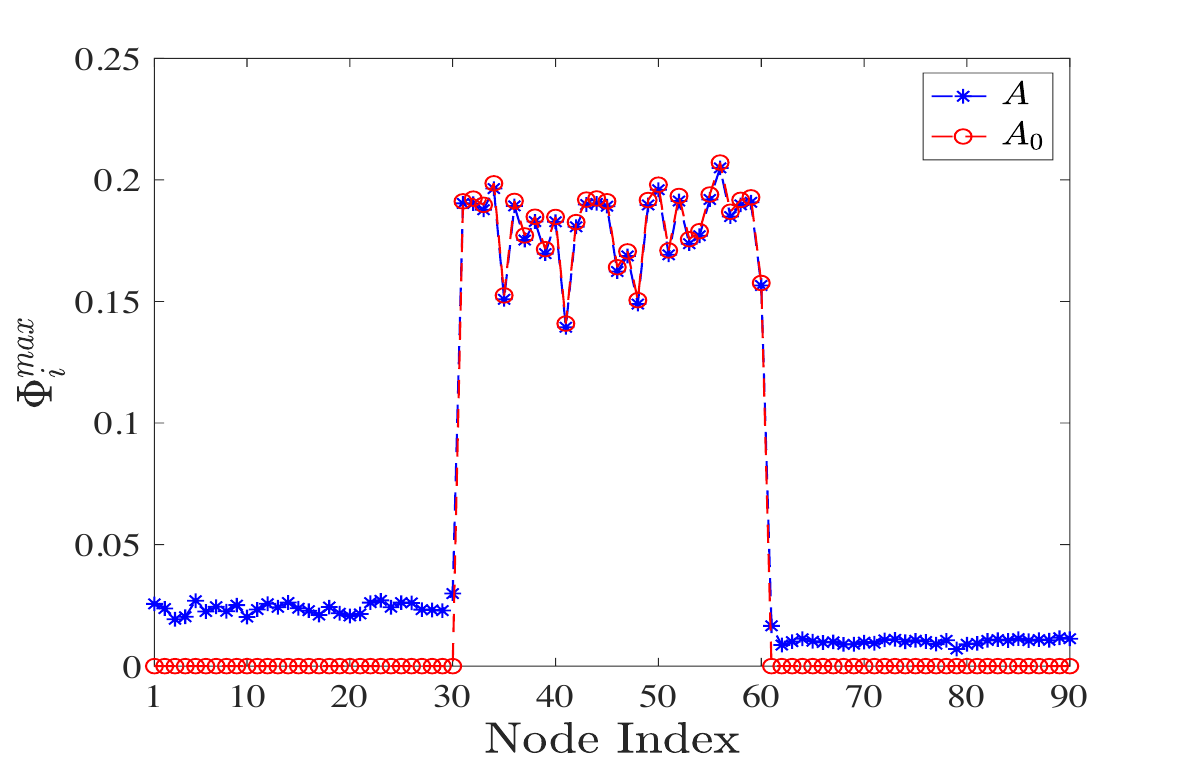}
\caption{\label{fig:figure7}Comparison between the two principal eigenvectors of \(\bm{A}\) and \(\bm{A}_0\). As can be noticed, two of the three blocks constituting \(\bm{A}_0\) are exactly zero.}
\end{figure}


\section{Upper and lower bounds for the $\lambda_A^{\max}$}

The numerical findings in the main text show that the Perron eigenvalue \(\lambda_A^{\max}\) is less than the mean degree \(\langle k \rangle\). To elucidate this phenomenon, we will first introduce two theorems and then provide a proof of the inequality that establishes upper and lower bounds for the Perron eigenvalue.

\begin{theorem}[Rayleigh Quotient~\cite{horn2012matrix}]
For a symmetric matrix \(\bm{A}\) and any non-zero vector \(\bm{v}\), the Rayleigh quotient satisfies:
\[
\frac{\bm{v}^T \bm{A} \bm{v}}{\bm{v}^T \bm{v}} \leq \lambda_A^{\max},
\]
where \(\lambda_A^{\max}\) is the largest eigenvalue of \(\bm{A}\). Equality holds when \(\bm{v}\) is the corresponding eigenvector \(\bm{\Phi}_A^{\max}\).
\label{th:Ray}
\end{theorem}

\begin{theorem}[Gershgorin Discs~\cite{horn2012matrix}]
Let \(\bm{A}\) be an \(n \times n\) complex matrix with entries \(A_{ij}\). For each \(i\), define the Gershgorin disc \(\mathcal{D}(A_{ii}, R_i)\) centered at \(A_{ii}\) with radius \(R_i = \sum_{j \neq i} |A_{ij}|\). Then, all eigenvalues \(\lambda\) of \(\bm{A}\) lie within the union of these \(n\) discs:
\[
\lambda \in \bigcup_{i=1}^{n} \mathcal{D}(A_{ii}, R_i)\,.
\]
\label{th:Gersh}
\end{theorem}

We now utilize these two theorems to establish upper and lower bounds for the Perron eigenvalue \(\lambda_A^{\max}\).

\begin{theorem}[Bounds on \(\lambda_A^{\max}\)]
\[
\langle k \rangle \leq \lambda_A^{\max} \leq k_{\max}
\]
\label{theorem:ineq}
\end{theorem}

\begin{proof}

The upper bound \(\lambda_A^{\max} \leq k_{\max}\) follows from the Gershgorin Discs Theorem \ref{th:Gersh}, as \(k_{\max}\) corresponds to the largest radius of the Gershgorin disc that contains all eigenvalues.

For the lower bound, consider the vector \(\bm{v} = \left(\frac{1}{\sqrt{N}}, \frac{1}{\sqrt{N}}, \ldots, \frac{1}{\sqrt{N}}\right)^T\). Applying the Rayleigh quotient \ref{th:Ray}, we obtain:

\[
\frac{\bm{v}^T \bm{A} \bm{v}}{\bm{v}^T \bm{v}} = \frac{\displaystyle\sum_{i,j} A_{ij} \frac{1}{N}}{1} = \frac{1}{N} \displaystyle\sum_{i} k_i = \langle k \rangle \leq \lambda_A^{\max}
\]
\end{proof}

This result is illustrated graphically in Fig. \ref{fig:Disc}.

\underline{Note}: For regular graphs, where \(\langle k \rangle = k_i\) for all \(i\), we have \(\lambda_A^{\max} = \langle k \rangle\), because in this case the mean degree \(\langle k \rangle\) equals \(k_{\max}\). In this scenario, the leading eigenvector \(\bm{v}_{\max}\) is also the principal eigenvector \(\bm{\Phi}_A^{\max}\), corresponding to the largest eigenvalue \(\lambda_A^{\max}\). The Rayleigh quotient for \(\bm{v}_{\max} = (1, 1, \ldots, 1)\) directly gives:

\[
\lambda_A^{\max} = \frac{\bm{v}_{\max}^T \bm{A} \bm{v}_{\max}}{\bm{v}_{\max}^T \bm{v}_{\max}} = \langle k \rangle.
\]

\begin{figure}[h!]
    \centering
\includegraphics[width=\linewidth]{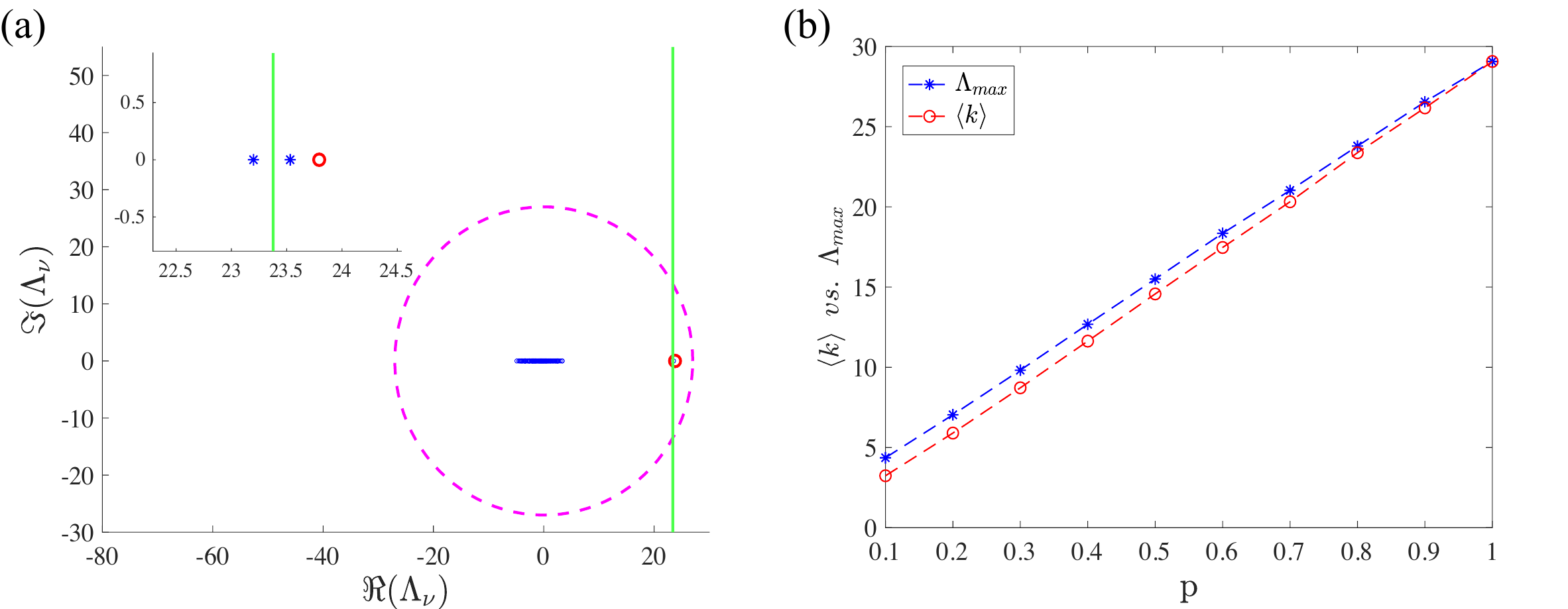}
    \caption{(a) The red point represents \(\Lambda_{max}\), the green line denotes \(\langle k \rangle\), and the blue points indicate the other eigenvalues. Here \( p=0.8 \) and from the inset it is clear that \(\Lambda_{max} < \langle k \rangle\). (b) illustrates that the mean degree is smaller than the largest eigenvalue, thereby supporting the Inequality result. As the probability $p$ of connection within each module increases, the difference between the largest eigenvalue and the mean degree diminishes, ultimately coinciding exactly when $p$ reaches 1. }
    \label{fig:Disc}
\end{figure}


{ \section{Mean-Field Proxy Challenging Topological Assumptions}

\subsection{Scale-free networks}

A \textit{scale-free network} is a network whose degree distribution follows a power law, \( P(k) \sim k^{-c} \), where \( c > 0 \) is typically between 2 and 3 in real-world networks \cite{newman_networks_2018}. One of the key features of scale-free networks is their broad asymmetric degree distribution, characterized by a few highly connected nodes, called \textit{hubs}, while the majority of nodes have relatively few connections. This structure makes scale-free networks highly heterogeneous, in contrast to Erdős–Rényi networks, where the degree distribution is much narrower, following a binomial (or Poisson) form, and most nodes have degrees close to the average \cite{newman_networks_2018}.

In the Barabási-Albert (BA) model, a widely used method for generating scale-free networks, the network grows by adding one node at a time, with each new node introducing \( m \) edges that connect to \( m \) existing nodes. These connections are established based on the mechanism of \textit{preferential attachment}, where nodes with higher degrees have a greater probability of receiving new connections. This process results in the emergence of hubs and a power-law degree distribution \( P(k) \sim k^{-c} \), with \( c = 3 \) in the standard BA model. In the following Fig. \ref{fig:App_SF}, we consider a three-module network, similar to the one described in the main text, with the difference that now each module is a scale-free subgraph generated using the Barabási-Albert model, and the modules are connected by single links.

\begin{figure}[h!]
    \centering
\includegraphics[width=\linewidth]{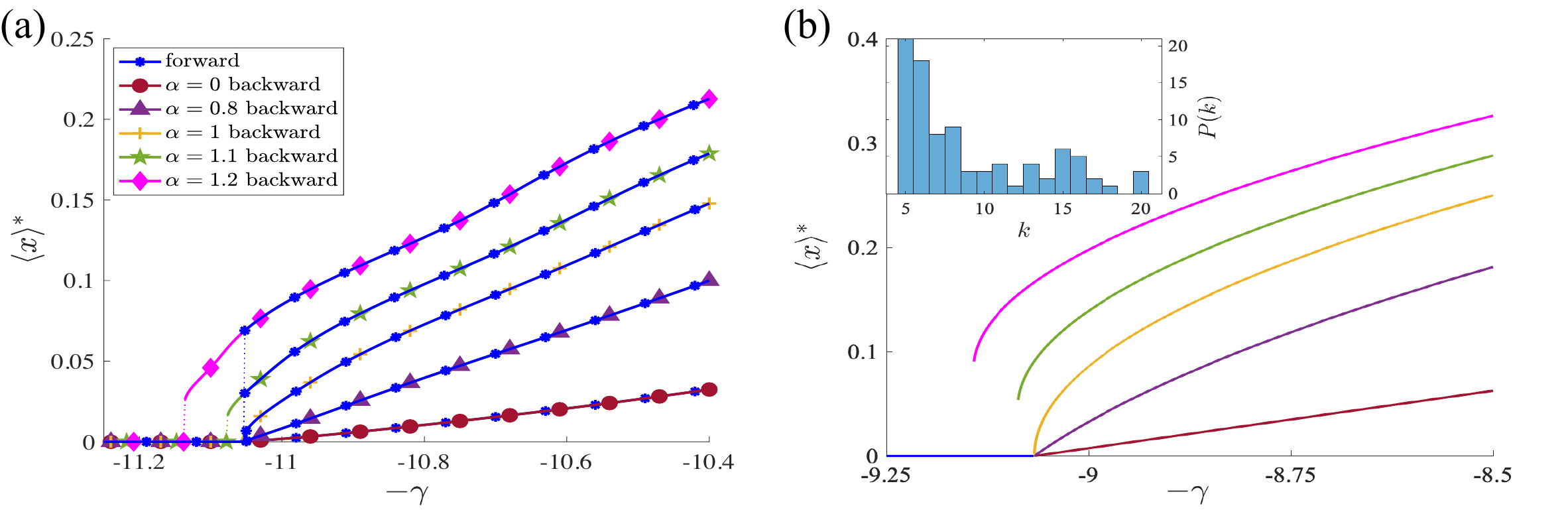}
    \caption{{(a) \textit{Numerical}: The blue lines represent the forward direction, while other colors represent the backward direction. For \(\alpha \leq 1\), the transitions are continuous, while for \(\alpha > 1\), they are discontinuous, characterized by abrupt changes and bistability. (b) \textit{Analytical}: Results confirm qualitative agreement with numerical findings. The inset shows the degree distribution, which follows a power-law with \(c = 3\). The network used is a modular network of 90 nodes, divided into 3 scale-free modules of 30 nodes each, with each module generated using the Barabási-Albert model with an attachment parameter \(m = 5\) and a degree distribution following a power law with exponent \(c = 3\). Single edges link the modules.}
    }
    \label{fig:App_SF}
\end{figure}

In this case, the numerical and analytical results exhibit very good qualitative agreement, despite significant differences in their critical points. This discrepancy arises because the mean degree of the scale-free network \(\langle k \rangle\) is substantially lower than its largest eigenvalue, due to the large gap between the average degree \(\langle k \rangle\) and the maximum degree \(k_{\max}\). Notably, this agreement holds even though the degree distribution is neither narrow nor symmetric, contradicting the second assumption. This behavior may be explained by the compensatory relationship \(x_i \sim k_i\), a degree-based assumption that works well for scale-free networks and is thus a hallmark of their structural dynamics \cite{pastor2001epidemic_1, pastor2001epidemic_2}.

\subsection{Small-world networks}

A \textit{small-world network} is characterized by high clustering coefficient and short path lengths, combining properties of regular lattices and random graphs. These networks are often used to model systems where local connections dominate but long-range links significantly reduce the average path length \cite{newman_networks_2018}.
The Watts-Strogatz (WS) model is a widely used generative model for small-world networks. It starts with a regular ring lattice where each node is connected to \(k\) (even) nearest neighbors, evenly split between both sides. Then, with probability \(p\), each edge is rewired to a random node. For \(p = 0\), the network remains a regular lattice; for \(p = 1\), it becomes a random graph. By tuning \(p\), the WS model interpolates between these extremes, capturing the transition from order to randomness and creating networks with small-world properties \cite{watts1998collective}. In the following, as the focus is on quantitatively investigating the validity of the mean-field approximation, we consider single-module small-world networks with 100 nodes.

To investigate the influence of \(p\) and \(k\) on the discrepancy between two modeling approaches—the \textit{Individual-Based Mean-Field (IBMF)}\footnote{The reason for referring to Eq. \eqref{eq:1} as a \textit{Mean-Field} is that it represents a mean-field approximation, which, in principle, should be derived from the particle dynamics described by the reaction through an averaging process over the master equation, thereby neglecting any stochastic effects.} and the \textit{Degree-Based Mean-Field (DBMF)}, introduced here as numerical and analytical references—we considered networks with \(N = 100\) nodes, set \(k = 4, 16, 30, 50\), and varied \(p\) from 0 to 1. In the \textit{DBMF} approach, the largest eigenvalue remains equal to \(k\), while in the \textit{IBMF} approach, the largest eigenvalue typically exceeds \(k\) as \(p\) increases. This divergence leads to differences in the critical points predicted by \textit{IBMF} and \textit{DBMF}, as shown in Fig. \ref{fig:App_WS} (a). 

\begin{figure}[h!]
    \centering
\includegraphics[width=\linewidth]{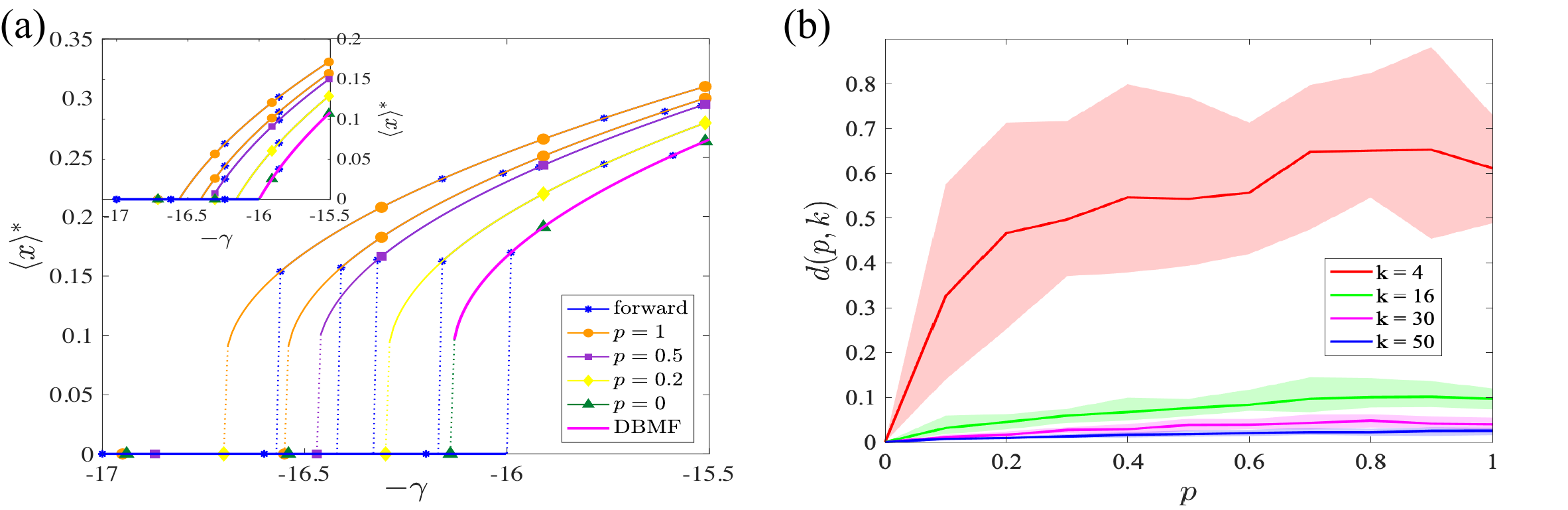}
    \caption{{(a) In the main panel, we set \(\alpha = 1.2\) and \(k = 16\). For \(p = 0, 0.2, 0.5\), we select examples where the error is closest to the average and plot the individual-based mean field. For \(p = 1\), we illustrate two cases: one with the maximum error and one with the minimum error. The case closer to the degree-based mean field corresponds to the smaller error. (b) We set \(N = 100\) and \(\alpha = 1.2\). The solid line represents the average error, while the shaded region indicates the error range. It can be observed that as \(k\) increases, the error decreases, and the range narrows.}
}
    \label{fig:App_WS}
\end{figure}

As seen in the scale-free case, the degree-based assumption \(x_i \sim k_i\), is a strong one, as we will show even for the small-world case. To further strengthen such assumption, the new opinion is seeded almost homogeneously across the network by perturbing the no-opinion fixed point. If this were not the case—such as in large rings where the opinion is initiated in a single node or individual—it would take significantly longer for the opinion to spread throughout the network, potentially affecting the accuracy of the mean-field proxy. In Fig. \ref{fig:App_WS} (a), we observed that when \(p = 0\), the IBMF and DBMF results are very close, which is the perfect scenario of the second assumption, which requires a narrow and symmetric degree distribution. This is because, in a Watts-Strogatz (WS) network with \(p = 0\), all nodes have the same degree and are thus equivalent, resulting in identical final states. Interestingly, as \(p\) increases, the second assumption diminishes, leading to a rise in the error. However, when \(p\) becomes sufficiently large, the WS network approaches the structure of a random graph, where better mixing compared to the original ring structure strengthens the first assumption. 

To quantify the error between the two methods, we first aligned the critical points by horizontally shifting the \textit{IBMF} curve to match the \textit{DBMF} critical point. The error was then calculated using \(d(p,k)=\sqrt{\sum_i \left(\text{IBMF}_i - \text{DBMF}_i\right)^2}\). As shown in Fig. \ref{fig:App_WS} (b), the error increases with \(p\), which is understandable given the diminishing the second assumption as \(p\) rises. This variability in results for the WS network when \(p = 1\) highlights the challenge of balancing the two assumptions: enhancing one may worsen the other. To address this for the WS model, we compensate for the weakening of second assumption at larger \(p\) by choosing a larger average degree, \(\langle k \rangle\). A higher \(\langle k \rangle\) keeps the degree distribution \(P(k)\) relatively narrow, while also improving mixing within the network, thereby strengthening the first assumption. This trade-off is particularly evident for \(4 < k \leq 30\), where the error initially increases but then shows a slight decline as \(p\) approaches 1, in line with the previous observation. In contrast, when \(k = 50\), this trend is less pronounced due to the small network size (\(N = 100\)), where even with a large \(p\), the structural changes in the graph remain limited, reducing the overall impact on the dynamics. Lastly, as expected for \(k = 4\), the error increases significantly with \(p\) and remains high across all values of \(p\), showing a slight decrease at \(p = 1\) but with considerable variability throughout.


\section{Polarization patterns beyond the critical point}

Previously, our main focus was on investigating the validity of the approximation and understanding the nature of phase transitions near the critical point, where the patterns were primarily determined by the most unstable eigenmode, resulting in relatively simple dynamics. In this section, we consider a richer scenario that emerges when the value of $\gamma$ is reduced further, allowing multiple eigenvalues of the Jacobian matrix to become positive. Specifically, for a network consisting of three modules, we anticipate that three eigenvalues will become significantly larger than the others. This leads to more intricate patterns, as the system's behavior becomes strongly influenced by the interplay of the corresponding eigenvectors. By extending the results numerically beyond this point, we aim to explore how these richer dynamics evolve and how polarization can be affected in the nonlinear regime under these conditions.

\begin{figure}[h!]
    \centering
\includegraphics[width=\linewidth]{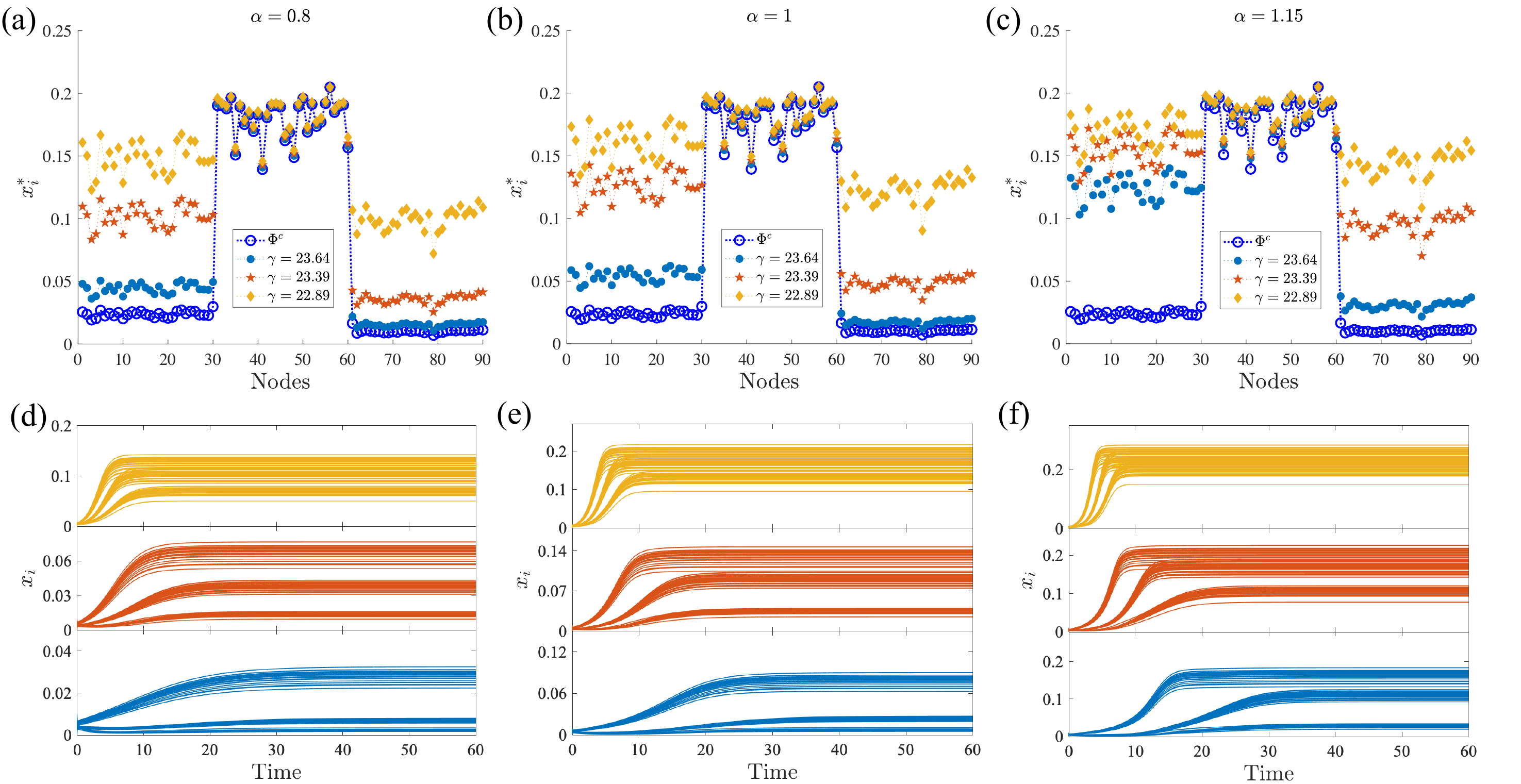}
    \caption{{\textit{Upper panels}: The panels (a), (b), and (c) correspond to \(\alpha = 0.8\), \(\alpha = 1\), and \(\alpha = 1.15\), with three different values of \(\gamma\): \(\gamma = 23.64\), where only one eigenvalue is positive; \(\gamma = 23.39\), where two eigenvalues are positive; and \(\gamma = 22.89\), where three eigenvalues are positive. As \(\alpha\) decreases, the distinction between the three modules becomes more pronounced, whereas increasing \(\alpha\) causes the modules to become more merged. The results are compared with the critical eigenvector corresponding to the largest eigenvalue. Additionally, polarization increases with decreasing \(\alpha\) and decreases with increasing \(\alpha\). \textit{Lower panels}: The panels (d), (e), and (f) are derived from (a), (b), and (c), respectively. Each panel uses the same color to represent the same \(\gamma\), illustrating the evolution of \(x_i\) over time. It is evident that as \(\gamma\) decreases and \(\alpha\) increases, polarization reduces.}}
    \label{fig:App_multi}
\end{figure}

In Fig. \ref{fig:App_multi}, the upper panels show the final patterns, while the lower panels display the temporal evolution. In the initial regime, the new opinion is distinctly separated across the three modules, driven by the eigenvector of the most unstable eigenmode. However, as the system reaches equilibrium in the nonlinear regime, this scenario changes drastically, depending on the values of $\alpha$. Notably, even for cases with continuous bifurcations (panels (a), (d); (b), (e)), the further the system starts from the critical value of $\gamma$, the more mixed the opinion range becomes among modules, leading to a loss of polarization.

}



\end{document}